\documentclass[12pt]{article} \usepackage{fullpage}
\usepackage{amssymb,amsmath}
\usepackage{lineno}

\usepackage{nicefrac}
\usepackage{url}
\usepackage{color}
\usepackage[american]{babel}
\usepackage{graphicx}
\usepackage{version}
\usepackage{subcaption}
\usepackage[textsize=footnotesize]{todonotes}

\newif\iffull
\fulltrue
\graphicspath{{figures/}}

\usepackage{amsthm}
\newtheorem{theorem}{Theorem}
\newtheorem{corollary}{Corollary}

\newtheorem{lemma}{Lemma}

\newcommand{\calL}{{\ensuremath{\cal L}}}
\newcommand{\calC}{{\ensuremath{\cal C}}}
\renewcommand{\L}{{\ensuremath{\mathtt L}}}
\newcommand{\flipL}{\protect\reflectbox{\L}}

\newcommand{\leaveout}[1]{}

\date{}
\title{Segment representations with small resolution}
\author{Therese Biedl\thanks{David R.~Cheriton School of Computer
Science, University of Waterloo, Waterloo, Ontario N2L 1A2, Canada.
Supported by NSERC.  The author would like to thank Daniel
Gon\c{c}alvez and Cornelia Schattman for helpful input.}
}

\begin{document}

\maketitle
\begin{abstract}
A {\em segment representation} of a graph is an assignment of
line segments in 2D to the vertices in such a way that two
segments intersect if and only if the corresponding vertices are adjacent.
Not all graphs have such segment representations, but they exist,
for example, for all planar graphs.

In this note, we study the resolution that can be achieved for
segment representations, presuming the ends of segments must
be on integer grid points.  We show that any planar graph
(and more generally, any graph that has a so-called $\L$-representation)
has a segment representation in a grid of width and height $4^n$.
\end{abstract}

%\linenumbers

%%%%%%%%%%%%%%%%%%%%%%%%%%%%%%%%%%%%%%%%%%%%%%%%%%%%%%%%%%%%%%%%%%%%%%%%
\section{Introduction}

\iffull
The problems {\sc String} and {\sc Segment} have received much
attention in recent history.  Here, {\sc String} refers to the
problem of, given a graph $G=(V,E)$, testing whether we can assign 
a curve $\calC(v)$ in 2D
to each vertex $v$ in such a way that $\calC(v)$ and $\calC(w)$ intersect
if and only if $(v,w)$ is an edge in the graph.  Such an assignment
is called a {\em string representation}.    This problem is
NP-hard \cite{Kra91} and also in NP \cite{SSS03}; the latter is
not at all obvious because string representations sometimes require
an exponential number of crossings among strings \cite{KM91}.

The {\sc Segment} problem is a special case of the {\sc String}
problem, where we want a {\em segment
representation}, i.e., $\calC(v)$ is a line
segment.  The NP-hardness proof of \cite{Kra91} also works for
segments (see \cite{Kra94}), so {\sc Segment} is also NP-hard but not known
to be in NP.  Indeed, the problem is unlikely to be in NP, since
it was shown to be hard for the class $\exists \mathbb{R}$, the
existential theory of the reals \cite{Mat14}, and it is widely assumed 
that $\exists \mathbb{R}$ is a strict superset of NP.  (See
e.g.~\cite{SchaeferBook} for more on $\exists \mathbb{R}$.)

For the purposes of displaying a graph via a segment representation,
we care not only about the existence, but also about the resolution
required to show it.  By moving endpoints slightly
one can always achieve rational coordinates at the endpoints 
without changing the represented graph.  After scaling 
therefore we may assume that endpoints have integer coordinates. 
We measure the {\em resolution of a segment representation} by the size
of the grid supporting all endpoints of the segments.

Assuming $\exists \mathbb{R}\neq$ NP, there are some graphs that
have a segment representation, but no segment-representation has
exponential resolution (i.e., $O(c^n)$ for some constant $c$, where
$n$ is the number of vertices).  
Namely, a segment-representation of exponential resolution
can be described using $O(poly(n))$ bits, and given such
a description, we can check in $O(poly(n))$ time whether this corresponds
to a segment representation of a graph.  If all graphs had a segment
representation of exponential resolution, therefore {\sc Segment}
would be in NP.

In his Ph.D.~thesis in 1984 \cite{Sch84}, Scheinermann famously asked whether
every planar graph (i.e., graph that can be drawn in the plane 
without crossing) has a segment representation.  This was proved
in 2009 by Chalopin and Gon\c{c}alvez \cite{CG09}; a second
(independent and much more accessible) proof was given by
Gon\c{c}alvez, Isenburg and Pennarun \cite{GIP17}.  In particular,
{\sc Segment} is {\em not} $\exists\mathbb{R}$-hard for planar
graphs (the answer is simply ``yes''), which raises the possibility
that they always have representations of exponential resolution.

Neither of \cite{CG09,GIP17} addressed the question of what resolution can
be achieved.   We show in this note
that the representation from \cite{GIP17} can be modified
to have resolution $4^n$.
Our result is not specific to the construction in \cite{GIP17},
but instead works for any graph that has an {\em $\L$-representation}.  This is
a string representation of a graph where every $\calC(v)$
has the shape of an $\L$, i.e., it consists of a horizontal and
a vertical segment that share their left/bottom endpoint.
It was shown by Middendorf and Pfeiffer \cite{MP92} that
every graph that has an $\L$-representation also has a
segment representation.  The approach of \cite{GIP17} is
hence to construct an $\L$-representation for any planar
graph and then to appeal to Middendorf and Pfeiffer's result.

It is not hard to show a bound of $O((2n)^{2n})$ on the resolution
achieved with the transformation by Middendorf and Pfeiffer (we
review this in Section~\ref{sec:LGamma}).  Our main result is
the following better bound:

\begin{theorem}
\label{thm:main}
Any graph that has an $\L$-representation (in particular therefore
any planar graph) has a segment representation with resolution $4^n$.
\end{theorem}

%The proof uses essentially the same algorithm as was proposed by
%Middendorf and Pfeiffer, but rather than incrementally adding one
%more segment at a time, we use a ``one-shot approach'' where the
%coordinates of all segments are computed immediately from the
%coordinates of the input $\L$-representation.

%%%%%%%%%%%%%%%%%%%%%%%%%%%%%%%%%%%%%%%%%%%%%%%%%%%%%%%%%%%%%%%%%%%%%%%%
\section{Proof of Theorem~\ref{thm:main}}

Assume that we are given an $\L$-representation of a graph.
We may assume that no endpoint of an $\L$ lies on another $\L$,
else we can lengthen its segment a bit.
We may assume also that all vertices have a non-zero length horizontal 
and vertical segment, else we can insert a very short one.  
Whether two $\L$'s intersect depends only on
the relative coordinates of segments; we may
therefore reassign endpoint-coordinates to integers in $\{0,\dots,2n-1\}$
in sorted order (breaking ties arbitrarily) and obtain the same graph.
Hence assume from now on that all endpoints of all segments of $\L$'s
have distinct coordinates in $\{0,\dots,2n-1\}$. 

It will be convenient to rotate the picture by $180^\circ$,
so that  $\calC(v)$ becomes an 
{\rotatebox[origin=c]{180}{\L}}.
We describe $\calC(v)$ by giving the four coordinates
$\ell,b,r,t\in \{1,\dots,2n\}$ of its bounding box
$[\ell,r]\times [b,t]$; thus $\calC(v)$ consists of the
top and right side of this box.

To define $s(v)$, first let $d(v)$ by the segment from $(0,2^t)$ to $(2^r,0)$,
i.e., the downward diagonal in rectangle $[0,2^r]\times [0,2^t]$.
Observe that $d(v)$ has slope $-2^{t-r}$.
To obtain $s(v)$, intersect $d(v)$ with the halfspaces $\{x\geq 2^\ell\}$
and $\{y\geq 2^b\}$.  The endpoints of $s(v)$ are hence  
$(2^\ell,2^t-2^{t-(r-\ell)})$ and $(2^r-2^{r-(t-b)},2^b)$, which have 
integral coordinates in $1,\dots,2^{2n-1}$.
See also Figure~\ref{fig:example}.

\begin{figure}[ht]
%\begin{figure}[p]
\hspace*{\fill}
\includegraphics[width=0.85\textwidth]{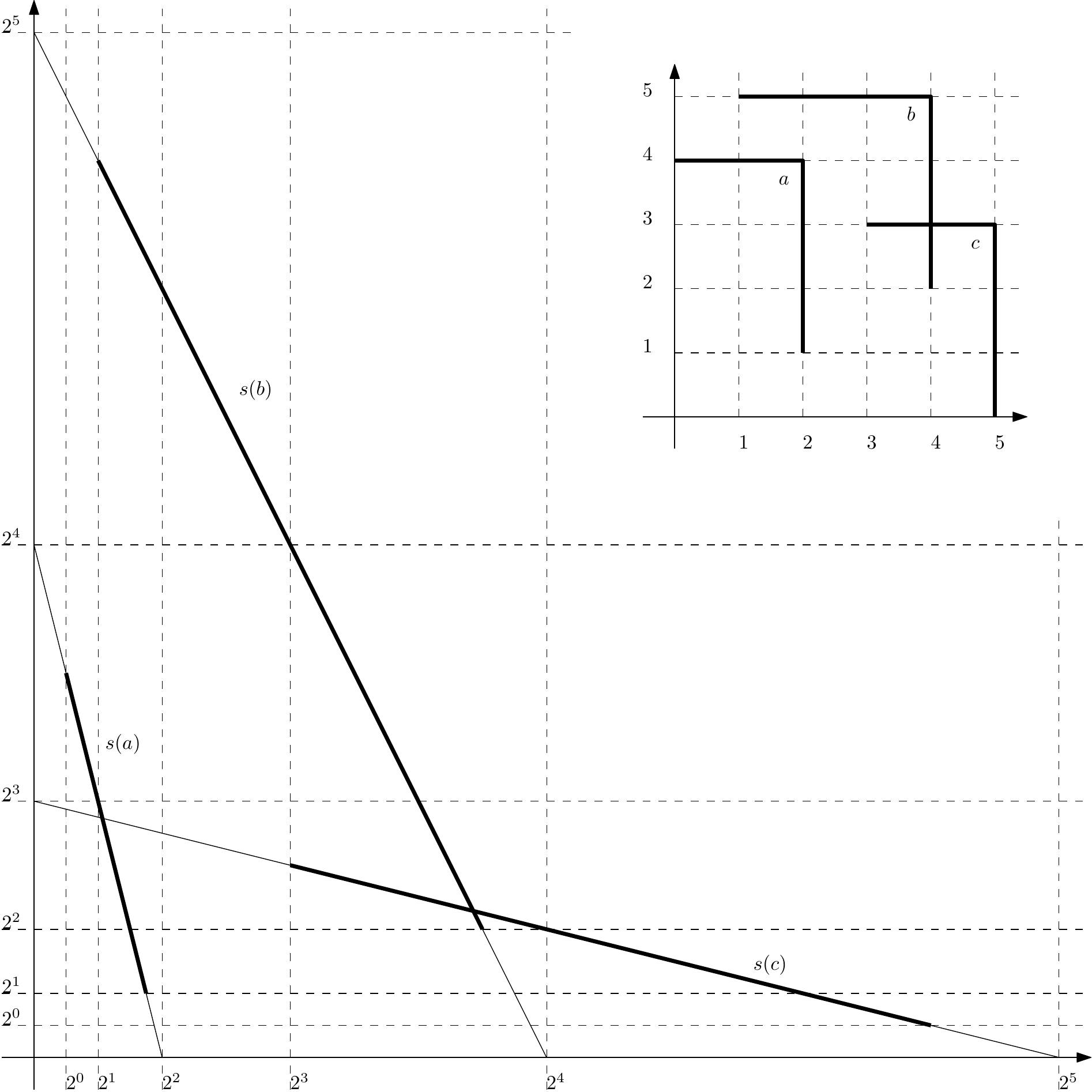}
\hspace*{\fill}
\caption{An $\L$-representation (rotated) and its corresponding segment representation.}
\label{fig:example}
\end{figure}

\begin{lemma}
\label{lem:intersect_L}
If $s(v)$ and $s(w)$ intersect, then $(v,w)$ is an edge.
\end{lemma}
\begin{proof}
Let $\calC(v)=\{\ell,b,r,t\}$ and $\calC(w)=\{L,B,R,T\}$.
After possible renaming we may assume $t>T$. 

Assume that $r>R$.  Walking along the $x$-axis, we encounter 
first $d(v)$ at $(2^r,0)$ and
then $d(w)$ at $(2^R,0)$.  Walking along the $y$-axis, we encounter
first $d(v)$ at $(0,2^t)$, and then $d(w)$ at $(0,2^T)$.  So the
order is the same, which means that line segments
$d(v)$ and $d(w)$ do not intersect, and neither can the
segments $s(v)$ and $s(w)$ that lie within them.
See e.g.~vertices $a$ and $b$ in Figure~\ref{fig:example}.

So we know that $r<R$ and $t>T$.  
If $r<L$ then $s(v)$ resides in the halfspace $\{x\leq 2^r\}$
and $s(w)$ resides in the halfspace $\{x\geq 2^L\}$ and they
do not intersect.    See e.g.~vertices $a$ and $c$ in
Figure~\ref{fig:example}.
If $b>T$, then similarly $s(v)$ resides in the halfspace $\{y\geq 2^b\}$
while $s(w)$ resides in the halfspace $\{y\leq 2^T\}$, and they
do not intersect.

So we know that $L<r<R$ and $b<T<t$.  But now consider 
point $(r,T)$.  This belongs to the vertical segment
of $v$ by $b<T<t$ and to the horizontal segment of $w$
by $L<r<R$.  So $\calC(v)$ and $\calC(w)$ intersect
and $(v,w)$ was an edge.
\end{proof}

\begin{lemma}
If $(v,w)$ is an edge, then $s(v)$ and $s(w)$ intersect.
\end{lemma}
\begin{proof}
Let $\calC(v)=\{\ell,b,r,t\}$ and $\calC(w)=\{L,B,R,T\}$.
After possible renaming we may assume 
that the intersection of the curves happens between
the vertical segment of $v$ and the horizontal segment of $w$.
See e.g.~vertices $b$ and $c$ in Figure~\ref{fig:example}.
In particular, $L<r<R$ and $b<T<t$, which by integrality means
$L+1\leq r \leq R-1$ and $b+1\leq T\leq t-1$.
This immediately implies
that $d(v)$ and $d(w)$ intersect (their order is different
along the $x$-axis and the $y$-axis).  The work is now to
show that this intersection point lies within the halfspaces
with which we pruned $d(v)$ and $d(w)$ to obtain $s(v)$ and 
$s(w)$.

We know that $d(v)$ has $y$-intercept $2^t$ and slope $-2^{t-r}$,
while $d(w)$ has $y$-intercept $2^T$ and slope $-2^{T-R}$.
If we let $(x^*,y^*)$ be the point of their intersection, we hence
have 
$ 2^t - 2^{t-r}x^* = y^* = 2^T - 2^{T-R}x^*,$ 
which implies
$$ x^* 
= \frac{2^t - 2^T}{2^{t-r}- 2^{T-r}}
\geq \frac{2\cdot 2^{t-1} - 2^{t-1}}{2^{t-r}} = 2^{r-1}.$$ 
So the intersection point lies within half-space $\{x\geq 2^{r-1}\}$,
which contains 
$\{x\geq 2^L\}$ since $r\geq L+1$
and 
$\{x\geq 2^\ell\}$ since $\ell<r$, hence $\ell\leq r-1$.

Symmetrically $y^*$ satisfies
$ 2^r - 2^{r-t} y^* = 2^R - 2^{R-T} y^*$
which shows that
$$ y^* = \frac{2^R - 2^r}{2^{R-T} - 2^{r-t}} \geq \frac{2^{R-1}}{2^{R-T}} = 2^{T-1}.$$ 
So the intersection point lies within half-spaces 
$\{y\geq 2^b\}$ and $\{y\geq 2^B\}$ by $T>B$ and $T>b$.

In consequence the intersection point of $d(v)$ and $d(w)$ lies
on both $s(v)$ and $s(w)$ and they intersect as desired.
\end{proof}

%%%%%%%%%%%%%%%%%%%%%%%%%%%%%%%%%%%%%%%%%%%%%%%%%%%%%%%%%%%%%%%%%%%%%%%%
\section{$\{\flipL,\L\}$-representations}
\label{sec:LGamma}

The result by Middendorf and Pfeiffer \cite{MP92} is more general
than stated before; they can create segment representations even
if the input representation is an $\{\flipL,\L\}$-representation, i.e., 
vertex-curves may also be horizontally reflected $\L$s.

We have not been able to generalize Theorem~\ref{thm:main}
to such representations, because its proof relies on that
we can distort both $x$-coordinates and $y$-coordinates
exponentially.  However, one can show that we can obtain
a segment representation of resolution $(2n)^{2n}$ (specifically,
its width is quite small at $2n$, but its height is $(2n)^{2n}$.

As before we will assume that the representation has distinct
coordinates in $\{0,\dots,2n-1\}$ and has been
rotated by 180$^\circ$.  So for any vertex $v$, 
$\calC(v)$ includes the top side of its
bounding box, and it includes either the left or the right side.

To define $s(v)$, set $Q(v)$ to be the rectangle $[\ell,r]\times [(2n)^b,(2n)^t]$,
and use as $s(v)$ the diagonal of $Q(v)$ that corresponds to the ends of
$\calC(v)$.  Put differently, if $\calC$ is 
${\rotatebox[origin=c]{180}{\L}}$ then $s(v)$ is the downward diagonal
of $Q(v)$, and 
if $\calC$ is ${\rotatebox[origin=c]{180}{\flipL}}$ then $s(v)$ is the upward
diagonal of $Q(v)$.  
See also Figure~\ref{fig:example_2}.
Clearly the ends of $s(v)$ are integer grid points
within the required range.
 
\begin{figure}[ht]
%\begin{figure}[p]
\hspace*{\fill}
\includegraphics[width=0.85\textwidth]{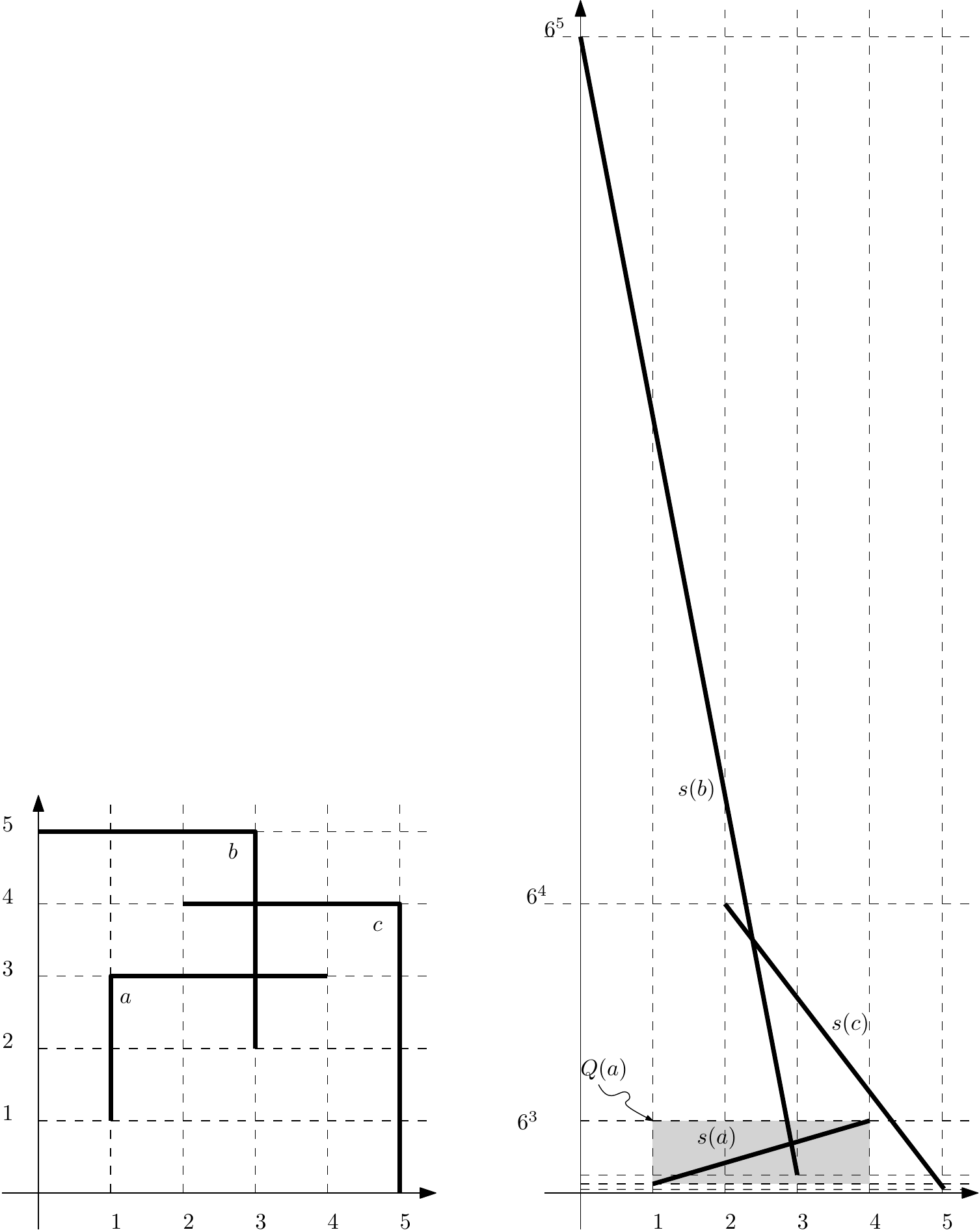}
\hspace*{\fill}
\caption{An $\{\L,\flipL\}$-representation (rotated) and its corresponding segment representation (not to scale).}
\label{fig:example_2}
\end{figure}

\begin{lemma}
\label{lem:intersect_Lflip}
If $s(v)$ and $s(w)$ intersect, then $(v,w)$ is an edge.
\end{lemma}
\begin{proof}
Let $\calC(v)=\{\ell,b,r,t\}$ and $\calC(w)=\{L,B,R,T\}$.
After possible renaming we may assume $t>T$.
Note that the constructed representation is horizontally symmetric,
so up to symmetry, we may assume that $\calC(v)$ is an
${\rotatebox[origin=c]{180}{\L}}$ and $s(v)$ is the downward diagonal
of $Q(v)$.

If $r<L$ then $s(v)$ resides in the halfspace $\{x\leq r\}$
and $s(w)$ resides in the halfspace $\{x\geq L\}$ and they
do not intersect.    
If $b>T$, then similarly $s(v)$ resides in the halfspace $\{y\geq 2^b\}$
while $s(w)$ resides in the halfspace $\{y\leq 2^T\}$, and they
do not intersect.
So we know that $L<r$ and $b<T<t$.

The horizontal segment of $\calC(w)$ is $(L,T)-(R,T)$.
If $R>r$, then this intersects 
the vertical segment $(r,b)-(r,t)$ of $\calC(v)$,
and $(v,w)$ is an edge and we are done.

The only remaining case is hence $\max\{L,R\}<r$ and $b<T<t$.
Recall that $s(w)$ resides within $Q(w)=[L,R]\times [(2n)^B,(2n)^T]$.
We claim that $s(v)$ does not intersect $Q(w)$.  Namely, 
$s(v)$ (which is the downward diagonal) has slope
$$-\frac{(2n)^t-(2n)^b}{r-\ell} < - \frac{(2n)^t}{2n} = - (2n)^{t-1} \leq - (2n)^T,$$
which means that at $x$-coordinate $r-1$, the $y$-coordinate of $s(v)$
is at least $(2n)^b+(2n)^T > (2n)^T$.  Since $r-1\geq R$, this means
that $s(v)$ bypasses $Q(w)$ entirely, and in particular $s(v)$
does not intersect $s(w)$.
\end{proof}

\begin{lemma}
If $(v,w)$ is an edge, then $s(v)$ and $s(w)$ intersect.
\end{lemma}
\begin{proof}
Let $\calC(v)=\{\ell,b,r,t\}$ and $\calC(w)=\{L,B,R,T\}$.
After possible renaming we may assume 
that the intersection of the curves happens between
the vertical segment of $v$ and the horizontal segment of $w$.
Up to symmetry, we may assume that $\calC(v)$ is an
${\rotatebox[origin=c]{180}{\L}}$, so $s(v)$ is the downward diagonal.
This implies $L<r<R$ and $b<T<t$.  

Consider the two vertical lines $\calL_{r{-}1}:\{x=r{-}1\}$ and $\calL_r:\{x=r\}$; 
these intersect both $s(v)$ and $s(w)$ by $\ell<r$ and $L<r<R$.
As in the previous proof, one shows that $s(v)$ has $y$-coordinate $>(2n)^T$
at $x=r-1$, while $s(w)$ has $y$-coordinate $\leq (2n)^T$ throughout,
so $s(w)$ intersects $\calL_{r-1}$ below $s(v)$.

Segment $s(v)$ intersects $\calL_r$ at $y$-coordinate $(2n)^b$ and ends here.  
(See e.g.~vertex $b$ in Figure~\ref{fig:example_2}.)
If $s(w)$ is the downward diagonal, then it intersects $\calL_r$ at
$y$-coordinate 
$$ (2n)^B+ (R{-}r) \frac{(2n)^T-(2n)^B}{R{-}L} 
= \frac{r{-}L}{R{-}L}(2n)^B + \frac{R{-}r}{R{-}L} (2n)^T
> \frac{1}{2n} (2n)^T = (2n)^{T{-}1} \geq (2n)^b,$$
so it is above $s(v)$ at $\calL_r$.
(See e.g.~vertex $c$ in Figure~\ref{fig:example_2}.)
If $s(w)$ is the upward diagonal, then it intersects $\calL_r$ 
at $y$-coordinate 
$$ (2n)^T- (R{-}r) \frac{(2n)^T-(2n)^B}{R{-}L} 
%> (2n)^T (1- \frac{R{-}r}{R{-}L}) + \frac{R{-}r}{R{-}L}  (2n)^B
=  \frac{r{-}L}{R{-}L} (2n)^T 
+ \frac{R{-}r}{R{-}L}(2n)^B 
> \frac{1}{2n} (2n)^{T} \geq (2n)^b,$$
so again it is above $s(v)$ at $\calL_r$.
(See e.g.~vertex $a$ in Figure~\ref{fig:example_2}.)
Either way therefore, somewhere
in the $x$-range $[r{-}1,r]$ the two line segments intersect.
\end{proof}

\section{Conclusion}

While we have made progress on one interesting question (what resolution
is needed for segment representations of planar graphs?), we leave many
questions open for future study:
\begin{itemize}
\item Is exponential resolution ever needed for planar graphs?  Put
	differently, can we construct a planar graph $G$ such that
	any segment representation of $G$ requires resolution
	$\Omega(c^n)$, for some constant $c$?  Or can we achieve
	polynomial resolution?

	The problem is also interesting 
	for the superclass of ``graphs that have an $\L$-representation''.
\item Can we achieve exponential resolution for all graphs with
	an $\{\reflectbox{\L},\L\}$-representation?
\item Can we construct an explicit graph family that has segment
	representations, but no such segment representation has
	exponential resolution?
\item Does every graph with a segment representation have one of
	resolution $f(n)$, for some computable function $f(n)$?
\end{itemize}

%%%%%%%%%%%%%%%%%%%%%%%%%%%%%%%%%%%%%%%%%%%%%%%%%%%%%%%%%%%%%%%%%%%%%%%%
\bibliographystyle{plain}
\bibliography{journal,full,gd,papers}

\begin{thebibliography}{10}

\bibitem{CG09}
J.~Chalopin and D.~Gon{\c{c}}alves.
\newblock Every planar graph is the intersection graph of segments in the
  plane: extended abstract.
\newblock In {\em {ACM} Symposium on Theory of Computing ({STOC} 2009)}, pages
  631--638, 2009.

\bibitem{GIP17}
Daniel Gon{\c{c}}alves, Lucas Isenmann, and Claire Pennarun.
\newblock Planar graphs as {L}-intersection or {L}-contact graphs.
\newblock In {\em {SIAM} Symposium on Discrete Algorithms (SODA 2018)}, pages
  172--184, 2018.

\bibitem{Kra91}
Jan Kratochv{\'{\i}}l.
\newblock String graphs {II.} {R}ecognizing string graphs is {NP}-hard.
\newblock {\em J. Comb. Theory, Ser. {B}}, 52(1):67--78, 1991.

\bibitem{Kra94}
Jan Kratochv{\'{\i}}l.
\newblock A special planar satisfiability problem and a consequence of its
  {NP}-completeness.
\newblock {\em Discrete Applied Mathematics}, 52(3):233--252, 1994.

\bibitem{KM91}
Jan Kratochv{\'{\i}}l and Ji\v{r}\'{i} Matou\v{s}ek.
\newblock String graphs requiring exponential representations.
\newblock {\em J. Comb. Theory, Ser. {B}}, 53(1):1--4, 1991.

\bibitem{Mat14}
Ji{\v{r}}{\'{\i}} Matou{\v{s}}ek.
\newblock Intersection graphs of segments and $\exists\mathbb{R}$.
\newblock {\em CoRR}, abs/1406.2636, 2014.

\bibitem{MP92}
Matthias Middendorf and Frank Pfeiffer.
\newblock The max clique problem in classes of string-graphs.
\newblock {\em Discrete Mathematics}, 108(1-3):365--372, 1992.

\bibitem{SSS03}
M.~Schaefer, E.~Sedgwick, and D.~Stefanovic.
\newblock Recognizing string graphs is in {NP}.
\newblock {\em Journal of Comput. Syst. Sci.}, 67(2):365--380, 2003.

\bibitem{SchaeferBook}
Marcus Schaefer.
\newblock {\em Crossing Numbers of Graphs}.
\newblock CRC Press, 2017.

\bibitem{Sch84}
Edward~R. Scheinerman.
\newblock {\em Intersection Classes and Multiple Intersection Parameters of
  Graphs}.
\newblock PhD thesis, Princeton University, 1984.

\end{thebibliography}

%%%%%%%%%%%%%%%%%%%%%%%%%%%%%%%%%%%%%%%%%%%%%%%%%%%%%%%%%%%%%%%%%%%%%%%%

\end{document}